\documentclass[submission,copyright,creativecommons]{eptcs}
\providecommand{\event}{ICE 2016} 
\usepackage{breakurl}             

\usepackage{amssymb}
\usepackage{amsmath}
\usepackage{stmaryrd}
\usepackage{prooftree}
\usepackage{amsthm}
\usepackage{pgf}
\usepackage{tikz}
\usepackage{verbatim}
\usetikzlibrary{arrows}
\usepackage{upgreek}
\usepackage{dsfont}
\usepackage{etoolbox}

\usepackage{appendix}
\capturecounter{theorem}
\capturecounter{lemma}

\newtoggle{tech_report}

\togglefalse{tech_report}

\iftoggle{tech_report}{

}{
}

\long\def\ignore#1{\relax}

\newcommand\struto[1][15pt]{{\raise #1 \hbox{\strut}}}%
\newcommand\strutb[1][15pt]{{\raise-#1 \hbox{\strut}}}%

 \newcommand\toaux[1]{\immediate\write\@auxout{#1}}

\newcommand\olditem{}
\newcommand\olditemize{}
\newcommand\oldenditemize{}
\newcommand\oldenumerate{}
\newcommand\oldendenumerate{}
\let\olditem\item
\let\olditemize\itemize
\let\oldenditemize\enditemize
\let\oldenumerate\enumerate
\let\oldendenumerate\endenumerate

\newcommand\myitem{}
\makeatletter
\def\myitem{\@ifnextchar[\@myitemwith\@myitemwithout}
\long\def\@myitemwith[#1]{\olditem[{#1}]\unskip}
\long\def\@myitemwithout{\olditem\unskip}
\makeatother

\renewenvironment{itemize}[1][0]{%
  \def\item{\removelastskip\myitem}%
  \removelastskip\olditemize\removelastskip}
{\removelastskip\oldenditemize\removelastskip%
  \def\item\olditem%
}

\renewenvironment{enumerate}[1][0]{%
  \def\item{\removelastskip\myitem}%
  \removelastskip\oldenumerate\removelastskip}
{\removelastskip\oldendenumerate\removelastskip%
  \def\item\olditem%
}

\newcommand\mybox[1]{\fbox{\vbox{#1}}}
\renewcommand\[[1][3]{\par\removelastskip\vskip#1pt\vbox\bgroup\hrule height0pt\vfil\hbox to\hsize\bgroup\hfil\(}
\renewcommand\][1][3]{\)\hfil\egroup\vfil\hrule height0pt\egroup\vskip#1pt\nointerlineskip\noindent}

\newbox\columnsbox
\newbox\tmpbox
\newdimen\columnsheight
\newdimen\columnwidth
\newdimen\remainingwidth
\newdimen\textwidthsave
\def\mycolumnsheight{}

\newcommand\columns[1]{%
  \def\mycolumnsheight{}%
  \setlength\remainingwidth\textwidth%
  \setbox\columnsbox=\vbox\bgroup\vskip0pt\vfil\hbox to\textwidth\bgroup#1\egroup\vfil\egroup%
  \columnsheight=\ht\columnsbox%
  \def\mycolumnsheight{to\columnsheight}%
  \hrule height 0pt\vtop{\hbox to\wd\columnsbox\bgroup#1\egroup}%
}

\makeatletter

\def\commonpart{%
  \setlength\columnwidth{\wd\tmpbox}%
  \vtop{\vskip0pt\hbox to\columnwidth{{\box\tmpbox}}}%
  \advance\remainingwidth-\columnwidth%
  \setlength\textwidth\textwidthsave%
  \hsize\textwidthsave%
}
\def\column{\unskip\setlength\textwidthsave\textwidth\@ifnextchar[\@columnwith\@columnwithout}
\long\def\@columnwith[#1]#2{%
  \def\newhsize{#1\dimexpr\textwidth\relax}%
  \hsize\newhsize%
  \ifdim\hsize<0.1pt\hsize\remainingwidth\fi%
  \setlength\textwidth\hsize%
  \setbox\tmpbox=\hbox to\hsize\bgroup\hfil\vtop\mycolumnsheight{\vskip0pt#2\vskip0pt}\hfil\egroup%
  \commonpart%
}
\long\def\@columnwithout#1{%
  \hsize\remainingwidth%
  \setlength\textwidth\hsize%
  \setbox\tmpbox=\hbox\bgroup\vtop\mycolumnsheight{\vskip0pt#1\vskip0pt}\egroup%
  \commonpart%
}
\makeatother

\newcommand{\eqdef}{:=\ }
\newcommand{\recdef}{::=\ }

\newcommand\mathFomega{F_\omega}
\newcommand\Fomega{\ifmmode\mathFomega\else$\mathFomega$\fi}
\newcommand\mathFomegaC{F_\omega^{\mathcal C}}
\newcommand\FomegaC{\ifmmode\mathFomegaC\else$\mathFomegaC$\fi}
\newcommand\mathDNE{\mathrm{DNE}}
\newcommand\DNE{\ifmmode\mathDNE\else$\mathDNE$\fi}

\newcommand\gL{\mathfrak{L}}

\newcommand{\LabelSend}[2]{({#1}, {#2})}
\newcommand{\SemLabel}[1]{\llbracket {#1} \rrbracket}

\newcommand{\ValDom}{\mathfrak{D}}
\newcommand{\ValPar}{\bot}
\newcommand{\SemDeri}[1]{\llbracket {#1} \rrbracket}
\newcommand{\RewOver}[1]{\xrightarrow {#1}}
\newcommand{\RewOverOpt}[2]{\xrightarrow {#2}_{#1}}
\newcommand{\SeqEmpty}{\varepsilon}
\newcommand{\SeqEq}{\asymp}
\newcommand{\Config}[2]{[{#1}] {#2}}

\newcommand{\RewList}{\leq}
\newcommand{\ValEx}{\ValDom^{c}}
\newcommand{\AHinst}[1]{H({#1})}
\newcommand{\IntA}[1]{\llbracket {#1} \rrbracket}
\newcommand{\TakeRole}[1]{\text{takerole}_{#1}}

\newcommand{\EquivCl}[1]{[{#1}]}

\newcommand{\subst}[3]{{#1} \{{#2} \eqdef {#3}\}}
\newcommand{\TRec}[2]{\text{rec }{#1} . {#2}}

\newtheorem{theory}{Theory}
\newtheorem{proposition}{Proposition}
\newtheorem{definition}{Definition}
\newtheorem{lemma}{Lemma}
\newtheorem{theorem}{Theorem}
\newtheorem{example}{Example}
\newtheorem{remark}{Remark}

\newcommand{\extenv}[3]{{#1}[{#2} \leftarrow {#3}]}

\newcommand{\AAssign}[3]{{#1} \leftarrow {#2}({#3})}
\newcommand{\ATest}[2]{({#1}) \in {#2} ?}
\newcommand{\ReadVar}[1]{\text{rv}({#1})}
\newcommand{\WriteVar}[1]{\text{wv}({#1})}

\iftoggle{tech_report}{
\title{A Modular Formalization of Reversibility\\ for Concurrent Models
  and Languages (TR)\thanks{This work was partially supported by the
  Italian MIUR project PRIN 2010-2011 CINA, the French ANR project REVER n.\ 
ANR 11 INSE 007 and the European COST Action IC1405.}
}
}{
\title{A Modular Formalization of Reversibility\\ for Concurrent Models
  and Languages\thanks{This work was partially supported by the
  Italian MIUR project PRIN 2010-2011 CINA, the French ANR project REVER n.\ 
ANR 11 INSE 007 and the European COST Action IC1405.}
}
}

\author{Alexis Bernadet
\institute{Dalhousie University, Canada}
\email{bernadet@lix.polytechnique.fr}       
\and
Ivan Lanese
\institute{Focus Team, University of Bologna/INRIA, Italy}
\email{ivan.lanese@gmail.com}
}

\def\titlerunning{A Modular Formalization of Reversibility for Concurrency}
\def\authorrunning{A. Bernadet \& I. Lanese}

\begin{document}

\maketitle

\begin{abstract}
Causal-consistent reversibility is the reference notion of reversibility for
concurrency.  We introduce a modular framework for defining causal-consistent
reversible extensions of concurrent models and languages.  We show how
our framework can be used to define reversible extensions of
formalisms as different as CCS and concurrent X-machines. The
generality of the approach allows for the reuse of theories and
techniques in different settings.
\end{abstract}

\section{Introduction}

\label{sec:Intro}

Reversibility in computer science refers to the possibility of
executing a program both in the standard forward direction, and
backward, going back to past states.
Reversibility appears in many settings, from the undo button present in most
text editors, to algorithms for rollback-recovery~\cite{Rollback}.
Reversibility is also used in state-space exploration, as in Prolog, or for
debugging~\cite{AkgulM04}.
Reversibility emerges naturally when modeling biological
systems~\cite{RevStru}, where many phenomena are naturally reversible, and in
quantum computing~\cite{QML}, since most quantum operations are reversible.
Finally, reversibility can be used to build circuits which are more energy
efficient than non-reversible ones~\cite{Landauer}.

Reversibility for concurrent systems has been tackled first
in~\cite{rccs}, mainly with biological motivations.
The standard definition of reversibility in a sequential setting, recursively
undo the last action, is not applicable in concurrent settings, where there
are many actions executing at the same time.
Indeed, a main contribution of~\cite{rccs} has been the definition of
\emph{causal-consistent} reversibility:
any action can be undone provided that all the actions depending on it (if any)
have already been undone. 
This definition can be applied to concurrent systems, and it is now a
reference notion in the field (non causal reversibility is also studied, e.g.,
in~\cite{PhillipsPathway}).
See~\cite{LMT:BEATCS} for a survey on causal-consistent reversibility.

Following~\cite{rccs}, causal-consistent reversible extensions of many
concurrent models and languages have been defined, using different
techniques~\cite{rhopi,revpi,revMuOz,revUlidowski,GiachinoLMT15,journal-rhopi}.
Nevertheless, the problem of finding a general procedure that given a
formalism defines its causal-consistent reversible extension is still
open: we tackle it here (we compare in Section~\ref{sec:Concl} with
other approaches to the same problem).  In more details, we present a
modular approach, where the problem is decomposed in three main steps.
The first step defines the information to be saved to enable
reversibility (in a sequential setting).  The second step concerns the
choice of the concurrency model used.  The last step automatically
builds a causal-consistent reversible extension of the given formalism
with the chosen concurrency model.

Our approach is not aimed at providing efficient (in terms of amount
of history information stored, or in terms of time needed to recover a
past state) reversible causal-consistent semantics, but at providing
guidelines to develop reversible causal-consistent semantics which are
correct by construction. Indeed, the relevant properties expected from
a causal-consistent reversible semantics will hold naturally because
of our construction. Also, it clarifies the design space of
causal-consistent reversibility, by clearly separating the sequential
part (step 1) from the part related to concurrency (step 2).

Hence, our approach can be used:
\begin{itemize}
\item in models and languages where one or more causal-consistent
  reversible semantics already exist (such as CCS, see
  Section~\ref{sec:RCCS}), to provide a reference model correct by
  construction, to compare against the existing ones and to classify
  them according to the choices needed in steps 1 and 2 to match them;
\item in models and languages where no causal-consistent reversible
  semantics currently exists (such as X-machines, see
  Section~\ref{sec:Automata}), to provide an idea on how such a
  semantics should look like, and which are the challenges to enable
  causal-consistent reversibility in the given setting.
\end{itemize}


Section~\ref{sec:IntroForm} gives an informal overview of our approach.
Section~\ref{sec:Formal} gives a formal presentation of the construction of a
reversible LTS extending a given one and proves that the resulting LTS
satisfies the properties expected from a causal-consistent reversible formalism.
In Section~\ref{sec:RCCS} we apply our approach to CCS.
In Section~\ref{sec:Automata} we show how to apply the same approach to systems
built around concurrent X-machines. Section~\ref{sec:Concl} compares with
related approaches and presents directions for future work.
\iftoggle{tech_report}{
Proofs missing from the main part are collected in
the Appendix.
}{
Missing proofs are collected in the companion technical report~\cite{TR}.
}

\section{Informal Presentation}

\label{sec:IntroForm}

We want to define a causal-consistent reversible
extension for a given formal model or language.  Assume that the model
or the language is formally specified by a calculus with terms $M$ whose
behavior is described by an LTS with transitions of the form:
\[M \RewOver u M'\]
To define its causal-consistent reversible extension using our
approach we need it to satisfy the following properties:
\begin{itemize}

\item The LTS is deterministic: If $M \RewOver u M_1$ and $M \RewOver u M_2$,
then $M_1 = M_2$.

\item The LTS is co-deterministic:
If $M_1 \RewOver u M'$ and $M_2 \RewOver u M'$, then $M_1 = M_2$.

\end{itemize}
In other words, the label $u$ should contain enough information on how
to go forward and backward. This is clearly too demanding, hence the
following question is natural:

\noindent\textbf{What to do if the LTS is not deterministic or not co-deterministic ?}

Note that this is usually the case. For example, in CCS~\cite{CCS}, a label is
either $\tau$, $a$ or $\overset{\_}{a}$, and we can have, e.g., $P \RewOver \tau
P_1$ and $P \RewOver \tau P_2$ with $P_1 \neq P_2$.

What we can do is to \emph{refine} the labels and, as a consequence,
the calculus, by adding information to them.  Therefore, if we have an
LTS with terms $M$ and labels $\alpha$ which is not deterministic or
not co-deterministic, we have to define:
\begin{itemize}

\item A new set of labels ranged over by $u$.

\item A new LTS with the same terms $M$ and with labels $u$ which is
deterministic and co-deterministic.

\item An interpretation $\SemDeri u = \alpha$ for each label $u$ such that:\\
$M \RewOver u M'$ iff $M \RewOver {\SemDeri u} M'$
(correctness of the refinement).
\end{itemize}
\begin{remark}[A Naive Way of Refining Labels]\mbox{}\\
\label{rmk:NaiveRefineLabel}
A simple way of refining labels $\alpha$ is as follows:
\begin{itemize}

\item Labels $u$ are of the form $(M, \alpha, M')$ for each $M$,
  $\alpha$ and $M'$ with $M \RewOver \alpha M'$.

\item We have only transitions of the form
$M_1 \RewOver {(M_1,\alpha,M_2)} M_2$.
This LTS is trivially deterministic and co-deterministic.

\item We define $\SemDeri {(M, \alpha, M')} = \alpha$.
The correctness of this refinement is trivial.
\end{itemize}
Therefore, it is always possible to refine an LTS to ensure determinism and
co-determinism.
Unfortunately, as we will see later, this way of refining is not suitable for
our aims, since we want some transitions, notably concurrent ones,
 to commute without changing their
labels, and this is not possible with the refinement above.

\end{remark}
Assume now that we have an LTS which is deterministic and co-deterministic
with terms $M$ and labels $u$, and a computation:
\[M_0 \RewOver {u_1} M_1 \ldots \RewOver {u_n} M_n\]
From co-determinism, if we only have the labels $u_1$, \ldots $u_n$ and the
last term $M_n$, we can retrieve the initial term $M_0$ and all the
intermediate terms $M_1 \ldots M_{n - 1}$.
As a consequence, we can use the following notation without losing
information:
\[M_0 \RewOver {u_1} \ldots \RewOver {u_n} M_n\]
Therefore, only the labels are needed to describe the history of a particular
execution that led to a given term.

Hence, to introduce reversibility it is natural to define configurations and
transitions as follows, where $(L\RewOver u)$ denotes the list obtained by appending $u$ after $L$:
\begin{itemize}

\item Configurations $R$ are of the form $(L, M)$ with $L$ a sequence of labels
$u_1, \ldots u_n$ such that there exists $M'$ with
$M' \RewOver {u_1} \ldots \RewOver {u_n} M$
(notice that $M'$ is unique).

\item Forward transitions:
If $M \RewOver u M'$, then $(L, M) \RewOver u ((L\RewOver u), M')$.

\item Backward transitions:
If $M \RewOver u M'$, then $((L\RewOver u), M') \RewOver {u^{-1}} (L, M)$.
\end{itemize}
In the LTS above, the terms are configurations $R$, and the
labels are either of the form $u$ (move forward) or $u^{-1}$ (move
backward).  This new formalism is indeed reversible. This can be proved
by showing that the Loop lemma~\cite[Lemma 6]{rccs}, requiring each
step to have an inverse, holds.
The main limitation of this way of introducing reversibility is that a
configuration can only do the backward step that cancels the last
forward step: If $R_1 \RewOver u R_2$ and $R_2 \RewOver {v^{-1}} R_3$,
then $u = v$ and $R_1 = R_3$.
This form of reversibility is suitable for a sequential setting, where
actions are undone in reverse order of completion. In a concurrent
setting, as already discussed in the Introduction, the suitable notion
of reversibility is causal-consistent reversibility, where any action
can be undone provided that all the actions depending on it (if any)
have already been undone. We show now how to generalize our model
so to obtain causal-consistent reversibility.

First, we require a symmetric relation $\ValPar$ on labels $u$.
Intuitively, $u \ValPar v$ means that the actions described by $u$ and
$v$ are independent and can be executed in any order. In concurrent
systems, a sensible choice for $\ValPar$ is to have $u \ValPar v$ if and only if
the corresponding transitions are concurrent. By choosing instead $u
\ValPar v$ if and only if $u=v$ we recover the sequential setting. Indeed,
causal-consistent reversibility coincides with sequential
reversibility if no actions are concurrent.

The only
property on $\ValPar$ that we require, besides being symmetric, is
the following one:

If $M_1 \RewOver u M_2$, $M_2 \RewOver v M_3$ and $u \ValPar v$, then
there exists $M_2'$ such that $M_1 \RewOver v M_2'$ and $M_2' \RewOver u M_3$
(\emph{co-diamond property}).

Thanks to this property, we can define an equivalence relation
$\SeqEq$ on sequences $L$ of labels: $L \SeqEq L'$ if and only if $L'$
can be obtained by a sequence of permutations of consecutive $u$ and
$v$ in $L$ such that $u \ValPar v$.

We can now generalize the definition of configuration $R = (L, M)$ by replacing
 $L$ by its equivalence class $\EquivCl L$ w.r.t.\ $\SeqEq$.
In other words, a configuration $R$ is now of the form $(\EquivCl L, M)$.
Actually, $\EquivCl L$ is a Mazurkiewicz trace~\cite{Maz88}.
Transitions are generalized accordingly.

For example, if $u \ValPar v$ we can have transition sequences such as:
\[(\EquivCl L, M_1) \RewOver u (\EquivCl {(L\RewOver u)}, M_2)
\RewOver v (\EquivCl{(L\RewOver u\RewOver v)}, M_3)
\RewOver {u^{-1}} (\EquivCl{(L\RewOver v)}, M_4)\]
because, since $(L\RewOver u\RewOver v) \SeqEq (L\RewOver v\RewOver u)$, we have
$\EquivCl{(L\RewOver u\RewOver v)} = \EquivCl{(L\RewOver v\RewOver u)}$.

We will show that the formalism that we obtain in this way is
reversible (the Loop lemma still holds) in a causal-consistent way.

Therefore, the work of defining a causal-consistent reversible
extension of a given LTS can be split into the three steps
below. 
\begin{enumerate}

\item\label{step:label} Refine the labels of the transitions.

\item\label{step:conc} Define a suitable relation $\ValPar$ on the newly defined
labels.

\item\label{step:construction} Define the configurations $R$ and the
  forward and backward transitions with the construction given above.

\end{enumerate}

Notice that step~\ref{step:label} depends on the semantics of
the chosen formalism and step~\ref{step:conc} depends on the chosen
concurrency model. These two steps are not
automatic. Step~\ref{step:construction} instead is a construction that
does not depend on the chosen formalism and which is completely
mechanical.  This modular approach has the advantage of allowing the
reuse of theories and techniques, in particular the ones referred to
step~\ref{step:construction}. Also, it allows one to better compare
different approaches, by clearly separating the choices related to the
concurrency model (step~\ref{step:conc}) from the ones related to the
(sequential) semantics of the formalism (step~\ref{step:label}).

Step~\ref{step:label} is tricky: we have to be careful when refining
the labels. We must add enough information to labels so that the LTS
becomes deterministic and co-deterministic.  However, the labels must
also remain unchanged when permuting two independent steps.  Therefore,
if we want to allow as many permutations as possible, we need to limit
the amount of information added to the labels.  For example, the
refinement given in Remark~\ref{rmk:NaiveRefineLabel}
only allows trivial permutations,
and in this case $\SeqEq$ can only be the identity.

\section{Introducing reversibility, formally}

\label{sec:Formal}

To apply the construction informally described in
Section~\ref{sec:IntroForm}, we need a formalism expressed as an LTS
that satisfies the properties of Theory~\ref{ax:Val}.

\begin{theory}
\label{ax:Val}
We have the following objects:
\begin{itemize}

\item A set $\ValDom$ of labels with a symmetric relation $\ValPar$ on $\ValDom$.

\item An LTS with terms $M$ and transitions $\RewOver u$ with
  labels $u \in \ValDom$.
\end{itemize}
The objects above satisfy the following properties:
\begin{itemize}
\item Determinism: If $M \RewOver u M_1$ and $M \RewOver u M_2$, then $M_1 = M_2$.

\item Co-determinism: If $M_1 \RewOver u M'$ and $M_2 \RewOver u M'$, then
$M_1 = M_2$.

\item Co-diamond property:
If $M_1 \RewOver u M_2$, $M_2 \RewOver v M_3$ and $u \ValPar v$, then
there exists $M_2'$ such that $M_1 \RewOver v M_2'$ and $M_2' \RewOver u M_3$.
\end{itemize}
\end{theory}
In the rest of this section, we assume to have the objects and
properties of Theory~\ref{ax:Val}.
For the formal definition of configurations we use
Mazurkiewicz traces~\cite{Maz88} (see Remark~\ref{rmk:Maz}, later on).
We give here a self-contained construction.

If we have the following sequence of transitions:
\[M_0 \RewOver {u_1} M_1 \RewOver {u_2} \ldots
\RewOver {u_{n - 1}} M_{n - 1} \RewOver {u_n} M_n\]
then the initial term $M_0$ and all the intermediate terms
$M_1$, \ldots $M_{n - 1}$ can be retrieved from
$u_1$, \ldots $u_n$ and $M_n$.
Therefore, by writing:
\[M_0 \RewOver {u_1} \RewOver {u_2} \ldots \RewOver {u_{n - 1}}
\RewOver {u_n} M_n\]
we do not lose any information.

We would like to manipulate formally transition sequences
as mathematical objects.
Hence it makes sense to use the following notation:
\begin{itemize}
\item A sequence $L = u_1$, \ldots $u_n$ of elements in $\ValDom$ is
  written $L = \RewOver{u_1} \ldots \RewOver{u_n}$.
\item The concatenation of $L_1$ and $L_2$ is written $L_1L_2$ and the
  empty sequence is written $\SeqEmpty$. $|L|$ denotes the length of
  sequence $L$.
\end{itemize}
Moreover, we want to consider sequences of
transitions up to permutations of independent steps.
\begin{definition}
\label{def:SeqEq}
The judgment $L \SeqEq L'$ is defined by the following rules:
\[\infer{L \SeqEq L}{} \qquad
\infer{L \SeqEq (L_1 \RewOver v \RewOver u L_2)}
{L \SeqEq (L_1 \RewOver u \RewOver v L_2) \quad u \ValPar v}
\]
\end{definition}
In other words, $L \SeqEq L'$ iff $L'$ can be obtained by doing
permutations of consecutive independent labels.
We can check that $\SeqEq$ satisfies the following properties:
\begin{toappendix}

\appendixbeyond 0

\begin{lemma}[Properties of $\SeqEq$]\strut

\label{lem:SeqEq}

\begin{enumerate}


\item $\SeqEq$ is an equivalence relation.

\item If $L_1 \SeqEq L_1'$ and $L_2 \SeqEq L_2'$, then
$(L_1 L_2) \SeqEq (L_1' L_2')$, that is relation $\SeqEq$ is closed under concatenation.

\item If $(L_1 \RewOver u) \SeqEq {L_2}$, then there exist $L_3$ and
$L_4$ such that $L_2 = (L_3 \RewOver u L_4)$, $L_1 \SeqEq (L_3 L_4)$ and for all
$v$ in $L_4$, $v \neq u$ and $u \ValPar v$.

\item If $(L_1 \RewOver u) \SeqEq (L_2 \RewOver u)$, then $L_1 \SeqEq L_2$.

\item If for all $v \in L$, $u \ValPar v$, then
$(L \RewOver u) \SeqEq (\RewOver u L)$.

\item If $L_1 \SeqEq L_2$, then $|L_1| = |L_2|$, that is relation $\SeqEq$ is length preserving. 
\end{enumerate}
\end{lemma}
\end{toappendix}
\begin{toappendix}[]

\begin{proof}
\mbox{}\\
\begin{enumerate}


\item
\begin{itemize}

\item Reflexivity: By definition.

\item Transitivity: We prove by induction on the derivation of $L_2
  \SeqEq L_3$ that if $L_1 \SeqEq L_2$ and $L_2 \SeqEq L_3$, then $L_1
  \SeqEq L_3$.

\item Symmetry: We prove by induction on the derivation of $L_1 \SeqEq
  L_2$ that if $L_1 \SeqEq L_2$, then $L_2 \SeqEq L_1$.  In
  particular, we use the definition of $\SeqEq$ and transitivity.

\end{itemize}

\item First, we prove by induction on the derivation of $L_1 \SeqEq
  L_2$ that if $L_1 \SeqEq L_2$, then $(L_3 L_1 L_4) \SeqEq (L_3 L_2
  L_4)$.  Then, from $L_1 \SeqEq L_1'$ and $L_2 \SeqEq L_2'$, we can
  deduce $(L_1 L_2) \SeqEq (L_1' L_2)$ and $(L_1' L_2) \SeqEq (L_1'
  L_2')$.  Hence, by transitivity $(L_1 L_2) \SeqEq (L_1' L_2')$.

\item By induction on the derivation of $(L_1 \RewOver u) \SeqEq L_2$.
Intuitively, every $v$ that moves to the right of the last $u$ must satisfy
$u \ValPar v$.

\item Corollary of item 3.

\item By induction on the length of $L$.

\item By induction on the derivation of $L_1 \SeqEq L_2$.

\end{enumerate}

\end{proof}

\end{toappendix}
We now define formally the judgment $M L M'$, representing a sequence
of transitions (also called a computation) with labels in $L$ starting
in $M$ and ending in $M'$, and prove some of its properties in
Lemma~\ref{lem:SeqVal}.
\begin{definition}
\label{def:SeqVal}
The judgment $M L M'$ is defined by the following rules:
\[\infer{M \SeqEmpty M}{} \qquad
\infer{M (L \RewOver u) M''}
{M L M' \quad M' \RewOver u M''}
\]
\end{definition}
\begin{toappendix}
\appendixbeyond 0
\begin{lemma}[Properties of $M L M'$]\strut

\label{lem:SeqVal}

\begin{enumerate}

\item The notation $M L M'$ is a conservative extension of the notation
$M \RewOver u M'$.

\item If $M L_1 M'$ and $M' L_2 M''$, then $M (L_1 L_2) M''$.

\item If $M_1 L M'$ and $M_2 L M'$, then $M_1 = M_2$.

\item If $M (L_1 L_2) M'$, then there exists $M''$ such that
$M L_1 M''$ and $M'' L_2 M'$.

\item If $M L M'$ and $L \SeqEq L'$, then $M L' M'$.
\end{enumerate}
\end{lemma}
\end{toappendix}
\begin{toappendix}[]
\begin{proof}
\mbox{}\\
\begin{enumerate}

\item Straightforward.

\item By induction on the derivation of $M' L_2 M''$.

\item By induction on the derivation of $M_1 L M'$.

\item By induction on the derivation of $M L_1 L_2 M'$.

\item By induction on the derivation of $L \SeqEq L'$ and by using item 4.

\end{enumerate}
\end{proof}
\end{toappendix}
The following lemma will be necessary to prove Theorem~\ref{th:RewConfSeq}:
\begin{toappendix}
\begin{lemma}
\label{lem:RewSeq}
If $(L_1 L_3) \SeqEq (L_2 L_4)$ and $(L_1 L_5) \SeqEq (L_2 L_6)$, then there
exist $L_7$, $L_8$, $L_3'$, $L_4'$, $L_5'$, $L_6'$ such that
$L_3 \SeqEq (L_7 L_3')$, $L_4 \SeqEq (L_8 L_4')$, $L_5 \SeqEq (L_7 L_5')$,
$L_6 \SeqEq (L_8 L_6')$, $L_3' \SeqEq L_4'$ and $L_5' \SeqEq L_6'$.
\end{lemma}
\end{toappendix}

\begin{toappendix}[
]

\begin{tikzpicture}[->,node distance=2.8cm, auto]

\node[]  (A) {};
\node[] (B) [above right of=A] {};
\node[] (C) [below right of=A] {};
\node[] (D) [right of=A] {};
\node[] (E) [right of=D] {};

\path (A) edge node {$L_1$} (B)
      (A) edge node {$L_2$} (C)
      (B) edge node {$L_3$} (D)
      (C) edge node {$L_4$} (D)
      (B) edge node {$L_5$} (E)
      (C) edge node {$L_6$} (E);

\end{tikzpicture}

\begin{tikzpicture}[->, node distance=2.8cm, auto]

\node (A) {};
\node (A1) [right of=A] {};
\node (A2) [right of=A1] {};
\node (B) [above right of=A] {};
\node (C) [below right of=A] {};
\node (D) [right of=B] {};
\node (E) [right of=C] {};
\node (F) [right of=A1] {};
\node (G) [right of=F] {};

\path (A) edge node {$L_1$} (B)
      (A) edge node {$L_2$} (C)
      (B) edge node {$L_7$} (D)
      (C) edge node {$L_8$} (E)
      (D) edge node {$L_3'$} (F)
      (E) edge node {$L_4'$} (F)
      (D) edge node {$L_5'$} (G)
      (E) edge node {$L_6'$} (G);

\end{tikzpicture}
\begin{proof}

We prove the result by induction on $|L_3| + |L_4| + |L_5| + |L_6|$.

Assume that $L_5$ is not empty.
Therefore, there exist $u$ and $L_7$ such that $L_5 = (L_7 \RewOver u)$. Let $n$
be the number of occurrences of $u$ in $L_1 L_5 = (L_1 L_7 \RewOver u)$. 
Since $L_2 L_6 \SeqEq L_1 L_5$, $L_2 L_6$ has $n$ occurrences of $u$ too. 

We have a case analysis according to whether the $n$-th occurrence
(from the left) of $u$ in $L_2 L_6$ is in $L_2$ or in $L_6$.

If it is in $L_6$, since $L_2 L_6 \SeqEq (L_1 L_7 \RewOver u)$,
$L_6$ is of the form $(L_8 \RewOver u L_9)$ and for every $v$ in $L_9$, $u \ValPar v$.
Therefore, $L_6 \SeqEq (L_8 L_9 \RewOver u)$, $(L_1 L_7 \RewOver u) \SeqEq (L_2 L_8 L_9 \RewOver u)$ and
$L_1 L_7 \SeqEq L_2 L_8 L_9$. 
We can use the induction hypothesis with $L_7$ instead of $L_5$ and $L_8 L_9$
instead of $L_6$ and conclude.

Now, if the $n$-th occurrence of $L_2 L_6$ is in $L_2$, we have to prove the
following intermediate result:

If $L_5 \SeqEq (L_8 \RewOver v L_9)$, this $v$ is the $m$-th occurrence of $v$ in
$(L_1 L_8 \RewOver v L_9)$,
and the $m$-th occurrence of $v$ in $L_2 L_6$ is in $L_2$, then
there exist $L_{10}$, $w$ and $k$ such that $L_5 \SeqEq (\RewOver w L_{10})$, this $w$
is the $k$-th occurrence of $w$ in $(L_1 \RewOver w L_{10})$ and the $k$-th occurrence of
$w$ in $L_2 L_6$ is in $L_2$.

This intermediate result is proved by induction on the length of $L_8$:
\begin{itemize}

\item If $L_8$ is empty we can conclude.

\item If $L_8 = (L_{10} \RewOver w)$, $w \ValPar v$ and $w \neq v$,
then $L_5 \SeqEq (L_{10} \RewOver v \RewOver w L_9)$.
By induction hypothesis with $L_{10}$ instead of $L_8$ and $(\RewOver w L_9)$ instead of
$L_9$, we can conclude.

\item If $L_8 = (L_{10} \RewOver w)$ and $w \ValPar v$ or $u = w$.
There exist $k$ such that this $w$ is the $k$-th $w$ in $(L_1 L_{10} \RewOver w \RewOver v L_9)$.
In every sequence equivalent to $(L_1 L_{10} \RewOver w \RewOver v L_9)$, the $k$-th occurrence of
$w$ is at the left of the $m$-th occurrence of $v$.
This is true, in particular, for $L_2 L_6$.
Since, the $m$-th occurrence of $v$ in $L_2 L_6$ is in $L_2$, so is the
$k$-th occurrence of $w$.
By induction hypothesis with $L_{10}$ instead of $L_8$, $w$ instead of $v$
and $(\RewOver v L_9)$ instead of $L_9$, we can conclude.

\end{itemize}

We can use the intermediate result to prove that there exist $L_8$, $v$ and $m$
such that $L_5 \SeqEq (\RewOver v L_8)$, this $v$ is the $m$-th occurrence of $v$ in
$(L_1 \RewOver v L_8)$ and the $m$-th occurrence of $v$ in $L_2 L_6$ is in $L_2$.

Therefore, $L_2$ has at least $m$ occurrences of $v$.
Hence, $L_2 L_4$ has at least $m$ occurrences of $v$.
Since, $L_2 L_4 \SeqEq L_1 L_3$, $L_1 L_3$ has at least $m$ occurrences of $v$.

If the $m$-th occurrence of $v$ in $L_1 L_3$ is in $L_1$, then $L_1$ has at
least $m$ occurrences of $v$ which is in contradiction with the fact that the
$m$-th occurrence of $v$ in $(L_1 \RewOver v L_8)$ is in not in $L_1$.

Therefore, the $m$-th occurrence of $v$ in $L_1 L_3$ must be in $L_3$.
Hence, $L_3$ is of the form $(L_9 \RewOver v L_{10})$.
\begin{itemize}

\item If, for every $w$ in $L_9$, $w \ValPar v$, then $L_3 \SeqEq (\RewOver v L_9 L_{10})$.
We can apply the induction hypothesis with $(L_1 \RewOver v)$ instead of $L_1$,
$L_8$ instead of $L_5$ and $L_9 L_{10}$ instead of $L_3$.

\item If not, there exist $w$ and $k$ such that we do not have $w \ValPar u$,
the $k$-th occurrence of $w$ in $L_1 L_3$ is in $L_3$ and at the left of the
$m$-occurrence of $v$.
Therefore, the $k$-th occurrence of $w$ in $L_2 L_4$ is at the left of the
$m$-th occurrence of $v$ which is in $L_2$.
Hence, the $k$-th occurrence of $w$ in $L_2 L_4$ is in $L_2$.
So, $L_2 L_6$ has at least $k$ occurrence of $w$ and the $k$-th occurrence of
$w$ in $L_2 L_6$ is at the left of the $m$-th occurrence of $v$.
Therefore, in $(L_1 \RewOver v L_8)$, the $k$-th occurrence of $w$ is at the left of the
$m$-th occurrence of $v$.
Hence, it is in $L_1$.
Therefore, $L_1$ has at least $k$ occurrences of $w$ which is in contradiction
with the fact that the $k$-th occurrence of $w$ in $L_1 L_3$ is in $L_3$.

\end{itemize}

Therefore, we have proved that if $L_5$ is not empty, we can conclude.
If $L_3$, $L_4$ or $L_6$ are not empty, we can do a similar proof.
If all $L_3$, $L_4$, $L_5$ and $L_6$ are empty, then the result is trivial.

Therefore, the proof of the lemma is complete.

\end{proof}
\end{toappendix}
Now, we have all the tools to define formally a new LTS reversible and causal
consistent extending the given one.
\begin{definition}[Reversible and Causal-Consistent LTS]
\label{def:Conf}
\mbox{}\\A configuration $R$ is a pair $(\EquivCl{L}, M)$ of a sequence $L$
modulo $\SeqEq$ and a term $M$, such that there exists $M'$ with
$M' L M$. We write $\Config L M$ for $(\EquivCl{L}, M)$.

The semantics of configurations is defined by the following rules:
\[\infer{\Config L M \RewOver u \Config {L \RewOver u} {M'}}{M \RewOver u M'}
\qquad
\infer{\Config {L \RewOver u} {M'} \RewOver {u^{-1}} \Config L M}{M \RewOver u M'}
\]

For a given configuration $R = \Config L M$, the unique $M'$ such that
$M' L M$ is independent from the choice of $L$ in the equivalence
class.  We call such $M'$ the \emph{initial term} of the configuration
$R$.
We also define the projection of a configuration on the last term as $\SemDeri {\Config L M} = M$.
\end{definition}
\begin{remark}
\label{rmk:Maz}
In the definition above, $\EquivCl L$ is a Mazurkiewicz Trace~\cite{Maz88}.
\end{remark}

The above definition is well posed:

\begin{lemma}
If $R$ is a configuration and $R \RewOver u R'$
(resp.\ $R \RewOver {u^{-1}} R'$) then $R'$ is a configuration.
Furthermore, $R$ and $R'$ have the same initial term.
\end{lemma}
\begin{proof}
By definition of the semantics of configurations.
\end{proof}
The calculus formalized in Definition~\ref{def:Conf} is a conservative
extension of the original one.  Indeed its forward transitions exactly
match the transitions of the original calculus:
\begin{theorem}[Preservation of the Semantics]\mbox{}\\
\begin{itemize}
\item If $R_1 \RewOver u R_2$, then $\SemDeri {R_1} \RewOver u \SemDeri {R_2}$.

\item If $\SemDeri {R_1} \RewOver u M'$, then there exists $R_2$ such that
$\SemDeri {R_2} = M'$ and $R_1 \RewOver u R_2$.
\end{itemize}
\end{theorem}
\begin{proof}
By definition of the semantics of configurations.
\end{proof}
We can also show that the calculus is reversible by proving that the
Loop Lemma~\cite[Lemma 6]{rccs} holds.
\begin{theorem}[Loop Lemma]
\label{th:LoopLemma}
$R \RewOver u R'$ iff $R' \RewOver {u^{-1}} R$.
\end{theorem}
\begin{proof}
By definition of the semantics of configurations.
\end{proof}
We finally need to prove that our formalism is indeed \emph{causal
  consistent}.  A characterization of causal consistency has been
presented in~\cite[Theorem 1]{rccs}. It requires that two
coinitial computations are cofinal iff they are equal up to causal
equivalence, where causal equivalence is an equivalence relation on
computations equating computations differing only for swaps of
concurrent actions and simplifications of inverse actions
(see Theorem~\ref{th:UsualCasualC} for a precise formalization).

Before tackling this problem we study when consecutive transitions can
be swapped or simplified.

\begin{toappendix}

\appendixbeyond 0

\begin{lemma}\label{lem:Conf}\mbox{}\\

\begin{enumerate}

\item If $R_1 \RewOver u R_2$, $R_2 \RewOver v R_3$ and $u \ValPar v$, then
there exists $R_2'$ such that $R_1 \RewOver v R_2'$ and $R_2' \RewOver u R_3$.

\item If $R_1 \RewOver {u^{-1}} R_2$, $R_2 \RewOver {v^{-1}} R_3$ and
$u \ValPar v$, then
there exists $R_2'$ such that $R_1 \RewOver {v^{-1}} R_2'$ and
$R_2' \RewOver {u^{-1}} R_3$.

\item If $R_1 \RewOver u R_2$ and $R_2 \RewOver {u^{-1}} R_3$, then
$R_1 = R_3$.

\item If $R_1 \RewOver {u^{-1}} R_2$ and $R_2 \RewOver u R_3$, then $R_1 = R_3$.

\item If $R_1 \RewOver u R_2$, $R_2 \RewOver {v^{-1}} R_3$ and $u \neq v$,
then there exists $R_2'$, such that $R_1 \RewOver {v^{-1}} R_2'$,
$R_2' \RewOver u R_3$ and $u \ValPar v$.

\end{enumerate}

\end{lemma}

\end{toappendix}

\begin{toappendix}[]

\begin{proof}
\mbox{}\\
\begin{enumerate}

\item Straightforward from the fact that sequences in configurations are
considered up to $\SeqEq$.

\item Corollary of the Loop lemma and the previous item.

\item There exist $L_1$, $L_2$, $L_3$, $M_1$, $M_2$ and $M_3$ such that
$R_1 = \Config {L_1} {M_1} $ and $R_2 = \Config {L_2} M_2$,
$R_3 = \Config {L_3} {M_3}$.
Therefore, $M_1 \RewOver u M_2$, $(L_1 \RewOver u) \SeqEq L_2$,
$M_3 \RewOver u M_2$ and $L_2 \SeqEq (L_3 \RewOver u)$.
Hence, $(L_1 \RewOver u) \SeqEq (L_3 \RewOver u)$.
Then, $M_1 = M_3$ and, by Lemma~\ref{lem:SeqEq}, item 4, $L_1 \SeqEq L_3$.
Therefore, $R_1 = R_3$.

\item There exist $L_1$, $L_2$, $L_3$, $M_1$, $M_2$ and $M_3$ such that
$R_1 = \Config {L_1} {M_1} $ and $R_2 = \Config {L_2} M_2$,
$R_3 = \Config {L_3} {M_3}$.
Therefore, $M_2 \RewOver u M_1$, $L_1 \SeqEq (L_2 \RewOver u)$,
$M_2 \RewOver u M_3$ and $(L_2 \RewOver u) \SeqEq L_3$.
Hence, $M_1 = M_3$ and $L_1 \SeqEq L_3$.
Therefore, $R_1 = R_3$.

\item There exist $L_1$, $L_2$, $L_3$, $M_1$, $M_2$ and $M_3$ such that
$R_1 = \Config {L_1} {M_1} $ and $R_2 = \Config {L_2} M_2$,
$R_3 = \Config {L_3} {M_3}$.
Therefore, $M_1 \RewOver u M_2$, $(L_1 \RewOver u) \SeqEq L_2$,
$M_3 \RewOver v M_2$ and $L_2 \SeqEq (L_3 \RewOver v)$.
Hence, $(L_1 \RewOver u) \SeqEq (L_3 \RewOver v)$.
By Lemma~\ref{lem:SeqEq}, item 3, there exist $L_4$ and $L_5$ such that
$(L_3 \RewOver v) = (L_4 \RewOver u L_5)$, $L_1 \SeqEq L_4 L_5$ and for
all $w  \in L_5$, $u \neq w$ and $u \ValPar w$.
If $L_5$ is empty, then $(L_3 \RewOver v) = (L_4 \RewOver u)$, and $u = v$.
This is a contradiction.
Therefore, $L_5$ is not empty and there exist $L_6$ and $w$ such that
$L_5 = (L_6 \RewOver w)$.
Hence, $(L_3 \RewOver v) = (L_4 \RewOver u L_6 \RewOver w)$.
Therefore, $w = v$, $L_3 = (L_4 \RewOver u L_6)$.
Moreover, $L_1 \SeqEq (L_4 L_5) = (L_4 L_6 \RewOver v)$ and
$R_1 = \Config {L_1} {M_1} = \Config {L_4 L_6 \RewOver v} {M_1}$ is a
configuration.
Hence, there exists $M_0$ such that $M_0 (L_4 L_6 \RewOver v) M_1$.
Therefore, there exists $M_2'$ such that $M_0 (L_4 L_6) M_2'$ and
$M_2' \RewOver v M_1$.
Hence, we have $\Config {L_4 L_6 \RewOver v} {M_1} \RewOver {v^{-1}}
\Config {L_4 L_6} {M_2'}$.
Let $R_2' \eqdef \Config {L_4 L_6} {M_2'}$.
Therefore, $R_1 \RewOver {v^{-1}} R_2'$.

Moreover, by the fact that $M_2' \RewOver v M_1$, $M_1 \RewOver u M_2$ and
$u \ValPar v$ (because $v \in L_5 = (L_6 \RewOver v)$), there exists $M_1'$ such
that $M_2' \RewOver u M_1'$ and $M_1' \RewOver v M_2$.
By the fact that we also have $M_3 \RewOver v M_2$, we have $M_1' = M_3$.
Hence, $M_2' \RewOver u M_3$.
Therefore, $\Config {L_4 L_6} {M_2'} \RewOver u
\Config {L_4 L_6 \RewOver u} M_3$.
For all $w \in L_6$, $w \in L_5 = (L_6 \RewOver v)$, and so $u \ValPar w$.
By Lemma~\ref{lem:SeqVal}, item 5, $(L_6 \RewOver u) \SeqEq (\RewOver u L_6)$.
Hence, $\Config {L_4 L_6 \RewOver u} {M_3} =
\Config {L_4 \RewOver u L_6} {M_3} =
\Config {L_3} {M_3} = R_3$.
Therefore, $R_2' \RewOver u R_3$.

\end{enumerate}

\end{proof}

\end{toappendix}

Given the previous result, we can define
(Definition~\ref{def:SeqConf}) a formal way to rearrange (like in the
original calculus) and simplify a sequence of transitions. Note that
each rule defining the transformation operator $\RewList$ (but for
reflexivity) is justified by an item of Lemma~\ref{lem:Conf}.  Since
some of these transformations are asymmetric, e.g., the simplification
of a step with its inverse, the resulting formal system is not an
equivalence relation but a partial pre-order.  For example, the
sequence $\RewOver u \RewOver {u^{-1}}$ can be transformed into
$\SeqEmpty$ but not the other way around.  The reason is that if $M
\RewOver u \RewOver {u^{-1}} M'$, then $M \SeqEmpty M'$.  However, we
may have $M \SeqEmpty M$ without necessarily having $M \RewOver u \RewOver
{u^{-1}} M'$ (in particular, if $M$ cannot perform $u$).

\begin{definition}
\label{def:SeqConf}
We write $\ValEx$ for the set of $\gamma$ of the form $u$ or $u^{-1}$.
Also, we define $(u^{-1})^{-1} = u$.
A sequence $\gL$ of elements $\gamma_1$, \ldots $\gamma_n$ is written
$\RewOver {\gamma_1} \ldots \RewOver {\gamma_n}$.
Also, we write $\gL^{-1}$ for $\RewOver {\gamma_n^{-1}} \ldots
\RewOver {\gamma_1^{-1}}$.

The judgment $R \gL R'$ is defined by the following rules:
\[\infer{R \SeqEmpty R}{} \qquad
\infer{R_1 (\gL \RewOver \gamma) R_3}{R_1 \gL R_2 \quad R_2 \RewOver \gamma R_3}
\]
The judgment $\gL \RewList \gL'$ is defined by the following rules:
\[\begin{array}{c}

\infer{\gL \RewList \gL}{} \qquad

\infer{\gL_1 \RewList (\gL_2 \RewOver v \RewOver u \gL_3)}
{\gL_1 \RewList (\gL_2 \RewOver u \RewOver v \gL_3) \quad u \ValPar v} \qquad

\infer{\gL_1 \RewList (\gL_2 \RewOver {v^{-1}} \RewOver {u^{-1}} \gL_3)}
{\gL_1 \RewList (\gL_2 \RewOver {u^{-1}} \RewOver {v^{-1}} \gL_3) \quad
u \ValPar v} \\

\infer{\gL_1 \RewList (\gL_2 \gL_3)}
{\gL_1 \RewList (\gL_2 \RewOver u \RewOver {u^{-1}} \gL_3)} \qquad

\infer{\gL_1 \RewList (\gL_2 \gL_3)}
{\gL_1 \RewList (\gL_2 \RewOver {u^{-1}} \RewOver u \gL_3)} \qquad

\infer{\gL_1 \RewList (\gL_2 \RewOver {v^{-1}} \RewOver u \gL_3)}
{\gL_1 \RewList (\gL_2 \RewOver u \RewOver {v^{-1}} \gL_3) \quad u \neq v}
\end{array}
\]
\end{definition}
We can now prove some properties of the judgments above.
\begin{toappendix}

\appendixbeyond 0
\begin{lemma}\label{lem:SeqConf}\mbox{}\\
\begin{enumerate}

\item The notation $R \gL R'$ is a conservative extension of the notation
$R \RewOver \gamma R'$.

\item If $R \gL_1 R'$ and $R' \gL_2 R''$, then $R (\gL_1 \gL_2) R''$.

\item If $R (\gL_1 \gL_2) R'$, then there exists $R''$ such that
$R \gL_1 R''$ and $R'' \gL_2 R'$.

\item $\RewList$ is a partial pre-order.

\item If $\gL_1 \RewList \gL_2$ and $R \gL_1 R'$, then $R \gL_2 R'$.

\item If $L_1 \SeqEq L_2$ then $L_1 \RewList L_2$ and
$L_1^{-1} \RewList L_2^{-1}$.
\end{enumerate}
\end{lemma}
\end{toappendix}
\begin{toappendix}[]
\begin{proof}
Follows from Lemma~\ref{lem:Conf} with the structure of proofs similar to
the ones found in Lemma~\ref{lem:SeqEq} and Lemma~\ref{lem:SeqVal}.
\end{proof}
\end{toappendix}
Using the transformations above, we can transform any transition
sequence into the form $L_1^{-1}L_2$ where $L_1$ and $L_2$ are
sequences of forward transitions: $L_1^{-1}L_2$ is composed by a
sequence of backward steps followed by a sequence of forward steps
(Theorem~\ref{th:RewConfSeq}). Intuitively, this means that any
configuration can be reached by first going to the beginning of the
computation and then going only forward. This result corresponds to
the one in~\cite[Lemma 10]{rccs}, which is sometimes called the
Parabolic lemma.  In addition, we show that if two computations are
coinitial and cofinal, then they have a common form $L_1^{-1}L_2$. As
we will see below, this is related to causal consistency \cite[Theorem
  1]{rccs}.

\begin{theorem}[Asymmetrical Causal Consistency]
\label{th:RewConfSeq}
If $R \gL_1 R'$ and $R \gL_2 R'$, then there exist $L_1$ and $L_2$
such that $\gL_1 \RewList L_1^{-1} L_2$ and $\gL_2 \RewList L_1^{-1} L_2$.
%
\end{theorem}
\begin{proof}
Suppose we have $R = \Config L M$ and $R' = \Config {L'} {M'}$.
By simplifying the occurrences of $\RewOver u \RewOver {u^{-1}}$
and by replacing occurrences of $\RewOver u \RewOver {v^{-1}}$ by
$\RewOver {v^{-1}} \RewOver u$ when $u \neq v$ in $\gL_1$, we can prove that
there exist $L_1$ and $L_2$ such that $\gL_1 \RewList L_1^{-1} L_2$.
Therefore, there exist a configuration $R_1 = \Config {L_3} {M_1}$ such that
$R L_1^{-1} R_1$ and $R_1 L_2 R'$.
Hence, $R_1 L_1 R$.
So, $L_3 L_1 \SeqEq L$ and $L_3 L_2 \SeqEq L'$.
Similarly, we can prove that there exist $L_4$, $L_5$ and $L_6$ such that
$\gL_2 \RewList L_4^{-1} L_5$, $L_6 L_4 \SeqEq L$ and $L_6 L_5 \SeqEq L'$.
Therefore, $L_3 L_1 \SeqEq L_6 L_4$ and $L_3 L_2 \SeqEq L_6 L_5$.
By Lemma~\ref{lem:RewSeq}, there exist $L_7$, $L_8$, $L_1'$, $L_4'$, $L_2'$,
$L_5'$ such that
$L_1 \SeqEq L_7 L_1'$, $L_4 \SeqEq L_8 L_4'$, $L_2 \SeqEq L_7 L_2'$,
$L_5 \SeqEq L_8 L_5'$, $L_1' \SeqEq L_4'$ and $L_2' \SeqEq L_5'$.
Hence, $\gL_1 \RewList 
L_1^{-1} L_2 \RewList L_1'^{-1} L_7^{-1} L_7 L_2' \RewList L_1'^{-1} L_2'$.
Similarly, $\gL_2 \RewList L_4'^{-1} L_5'$.
We also have $L_1'^{-1} L_2' \SeqEq L_4'^{-1} L_5'$.
Therefore, $\gL_1 \RewList L_1'^{-1} L_2'$ and $\gL_2 \RewList L_1'^{-1} L_2'$.
\end{proof}
\begin{remark}
Our Theorem~\ref{th:RewConfSeq} could be stated with the terminology in
\cite{rccs} as follows:

If $s_1$ and $s_2$ are coinitial and cofinal computations, then
there exists $s_3$ which is a simplification of both $s_1$ and $s_2$.

We will show below that this is stronger than the implication "if two
computations are coinitial and cofinal then they are causal
equivalent" in the statement of causal consistency \cite[Theorem
  1]{rccs}.  Moreover, the (easier) implication "if two computations
are causal equivalent than they are coinitial and cofinal" of
\cite[Theorem 1]{rccs} is also true by construction in our framework.
\end{remark}
%

In order to define causal equivalence we need to restrict to valid
reduction sequences. This is needed since otherwise the
transformations defining causal equivalence, differently from the ones
defining $\RewList$, may not preserve validity of reduction sequences.
The usual definition of a valid reduction sequence of length $n$ is described
by $(n + 1)$ configurations $(R_i)_{0\leq i \leq n}$ and $n$ labels
$(\gamma_i)_{1 \leq i \leq n}$ such that:
\[R_0 \RewOver {\gamma_1} \ldots \RewOver {\gamma_n} R_n\]
However, by determinism and co-determinism, we can retrieve $R_0$, \ldots $R_{n-1}$ from
$\gamma_1$, \ldots $\gamma_n$ and $R_n$.
Therefore, we will use the following equivalent definition:
\begin{definition}[Valid Sequences]
A valid sequence $s$ is an ordered pair $(\gL, R)$ such that there exists
$R'$ with $R' \gL R$.
Then, $R'$ is unique, $R'$ and $R$ are called the initial and
final configuration of $s$, and we write $R' s R$.
We write $E$ for the set of valid sequences.
\end{definition}
Following \cite{rccs}, we can define \emph{causal equivalence}
$\thicksim$ as follows: if $s_2$ is a rewriting of $s_1$ (valid
permutation, or valid simplification), then $s_1 \thicksim s_2$.
More formally:
\begin{definition}[Causal Equivalence]
Causal equivalence $\thicksim$ is the least equivalence relation on
$E$ closed under composition satisfying the following equivalences
(provided both the terms are valid):
\begin{itemize}
\item we can swap independent actions: if
$u \ValPar v$ then $\RewOver u \RewOver v\thicksim \RewOver v \RewOver u$, $\RewOver {u^{-1}} \RewOver{v^{-1}} \thicksim \RewOver{v^{-1}}\RewOver{u^{-1}}$ and $\RewOver u \RewOver {v^{-1}} \thicksim \RewOver {v^{-1}} \RewOver u$;  
\item we can simplify inverse actions: $\RewOver u \RewOver {u^{-1}} \thicksim \SeqEmpty$ and $\RewOver {u^{-1}} \RewOver u \thicksim \SeqEmpty$.
\end{itemize}
\end{definition}
\begin{theorem}[Causal Consistency]
\label{th:UsualCasualC}
Assume that $s_1$ and $s_2$ are valid sequences.
Then we have $s_1 \thicksim s_2$ if and
only if $s_1$ and $s_2$ are coinitial and cofinal.
\end{theorem}
\begin{proof}\mbox{}\\
\begin{itemize}

\item If $s_1 \thicksim s_2$, by construction each step in the derivation of
$s_1 \thicksim s_2$ does not change the initial and final terms.
Therefore, $s_1$ and $s_2$ are coinitial and cofinal.

\item

Assume that $s_1$ and $s_2$ are coinitial and cofinal.
We want to prove that $s_1 \thicksim s_2$.

First, by using Lemma~\ref{lem:SeqConf} item 5,
for every $s = (\gL, R) \in E$ and $\gL'$,
if $\gL \RewList \gL'$ then we have $s' \in E$
and $s \thicksim s'$, where $s'$ stands for $(\gL', R)$.

By hypothesis, there exist $R$, $R'$, $\gL_1$ and $\gL_2$ such that
$s_1 = (\gL_1, R)$, $s_2 = (\gL_2, R)$, $R' \gL_1 R$ and $R' \gL_2 R$.

By Theorem~\ref{th:RewConfSeq}, there exists $\gL_3$ such that
$\gL_1 \RewList \gL_3$ and $\gL_2 \RewList \gL_3$.

Therefore, if we write $s_3$ for $(\gL_3, R)$, then $s_3 \in E$,
$s_1 \thicksim s_3$ and $s_2 \thicksim s_3$.
Hence, by the fact that $\thicksim$ is an equivalence relation, we have
$s_1 \thicksim s_2$.
\end{itemize}
\end{proof}

\section{Making CCS Reversible}
\label{sec:RCCS}
In this section we give a refinement of CCS with recursion so that we can apply
the framework described in Section~\ref{sec:Formal}. See \cite{CCS} for
details of CCS.

Assume that we have a set of channels $a$ and a set of process variables $X$.
In our refinement, as specified in Section~\ref{sec:Formal}, terms are
defined by the same grammar used in (standard) CCS:
\[\begin{array}{lll}
P, Q & \recdef & \underset{i \in I}{\Sigma} \alpha_i . P_i ~ \mid ~
               (P | Q) ~ \mid ~
               \nu a . P ~ \mid ~
                0 ~ \mid ~
                X ~ \mid ~
                \TRec X P \\
\alpha & \recdef & a ~ \mid ~ \overline a 
\end{array}
\]
Action $a$ denotes an input on channel $a$, while $\overline a$ is
the corresponding output. Nondeterministic choice $\underset{i \in
  I}{\Sigma} \alpha_i . P_i$ can perform any action $\alpha_i$ and
continue as $P_i$. $P|Q$ is parallel composition. Process $\nu a . P$
denotes that channel $a$ is local to $P$. $0$ is the process that does
nothing. Construct $\TRec X P$ allows the definition of recursive
processes.  Transitions in CCS are of the form $P
\RewOverOpt{ccs}{\alpha} P'$ where $\alpha$ is $a$, $\overline a$
or $\tau$ (internal synchronization).

The following grammar defines labels for our refinement:
\[u, v \recdef \LabelSend {(\alpha_j, P_j)_{j \in I}} i ~ \mid ~
              (u | \bullet) ~ \mid ~ (\bullet | u) ~ \mid ~ (u | v) ~ \mid ~
              \nu a . u ~ \mid ~ \TRec X P
\]
We consider terms and labels up to $\alpha$-equivalence (of
both variables $X$ and channels $a$).
Therefore, we can define the substitution $\subst P X Q$ avoiding variable
capture.
We define the interpretation of labels as follows (partial function):
\[\begin{array}{llllllll}
\SemLabel {\LabelSend {(\alpha_j, P_j)_{j \in I}} i} & \eqdef & \alpha_i & \qquad
\SemLabel {\TRec X P} & \eqdef & \tau\\

\SemLabel {(u | \bullet)} & \eqdef & \SemLabel u & \qquad

\SemLabel {(\bullet | u)} & \eqdef & \SemLabel u \\

\SemLabel {(u | v)} & \eqdef & \tau & \qquad

\SemLabel {\nu a . u} & \eqdef & \SemLabel u \quad
(\SemLabel u \notin \{a, \overline a\})
\end{array}
\]
We define transitions with the rules in Table~\ref{fig:RCCS}.
\begin{table}[t!]
\[\begin{array}{|c|}
\hline\\
\infer{\underset{j \in I}{\Sigma} \alpha_j . P_j \RewOver 
{\LabelSend {(\alpha_j, P_j)_{j \in I}} i} P_i}{i \in I} \qquad

\infer{(P | Q) \RewOver {(u | \bullet)} (P' | Q)}{P \RewOver u P'} \\[3mm]

\infer{(P | Q) \RewOver {(\bullet | u)} (P | Q')}{Q \RewOver u Q'} \qquad

\infer{\nu a . P \RewOver {\nu a . u} \nu a . P'}
{P \RewOver u P' \quad \SemLabel u \notin \{a, \overline{a}\}}\\[6mm]

\infer{(P | Q) \RewOver {(u | v)} (P' | Q')}
{P \RewOver u P' \quad Q \RewOver v Q' \quad
(\SemLabel u = \alpha \wedge \SemLabel v = \overline{\alpha}) \vee
(\SemLabel u = \overline{\alpha} \wedge \SemLabel v = \alpha)} \\[6mm]

\infer{\TRec X P \RewOver {\TRec X P} \subst P X {\TRec X P}}{}\\[3mm]
\hline
\end{array}
\]
\caption{Refined transitions for CCS}
\label{fig:RCCS}
\end{table}

\begin{proposition}
$P \RewOver u P'$ iff $\SemLabel u$ exists and
$P \RewOverOpt {ccs} {\SemLabel u} P'$.
\end{proposition}
\begin{proof}
By induction on $P \RewOver u P'$ for the forward implication, and by
induction on $P \RewOverOpt {ccs} {\SemLabel u} P'$ for the backward
implication: each rule here corresponds to a rule in the semantics of
CCS~\cite{CCS}.
\end{proof}
Now we only have to define a suitable $\ValPar$. Below, $\upxi$
stands for $u$ or $\bullet$.
\[\begin{array}{c}

\infer{u \ValPar \bullet}{} \qquad

\infer{\bullet \ValPar u}{} \qquad

\infer{(\upxi_1 | \upxi_2) \ValPar (\upxi_1' | \upxi_2')}
{\upxi_1 \ValPar \upxi_1' \quad \upxi_2 \ValPar \upxi_2'} \qquad

\infer{\nu a . u \ValPar \nu a . v}{u \ValPar v}

\end{array}
\]

Informally, $u \ValPar v$ means that the transitions described by $u$
and $v$ operate on separate processes. Note that it would not be
possible to define $\ValPar$ on the original CCS labels, since they do
not contain enough information. The refinement of labels solves this
problem too.

\begin{toappendix}

\appendixbeyond 0

\begin{theorem}\label{th:rccs}

The LTS and the relation $\ValPar$ defined for CCS satisfy Theory~\ref{ax:Val}.

\end{theorem}

\end{toappendix}

\begin{toappendix}[]

\begin{proof}
Determinism and co-determinism are straightforward: the proof is by induction on
the label, and we can notice that for each rule, from the label and a
term, we have enough information to deduce the other term.

The proof of the co-diamond property is by induction on the first
label.
\end{proof}

\end{toappendix}

Thanks to this result, we can apply the framework of
Section~\ref{sec:Formal} to obtain a causal-consistent reversible
semantics for CCS. We also have for free results such as the Loop
lemma or causal consistency.

\begin{example}
Consider the CCS process $a.b.0|\overline{b}.c.0$.
We have, e.g., the two computations below:\\
$a.b.0|\overline{b}.c.0 
\RewOver {\LabelSend {(a, b.0)} 1 | \bullet} 
b.0|\overline{b}.c.0 
\RewOver {\bullet | \LabelSend {(\overline{b}, c.0)} 1}
b.0|c.0 
$\\
$a.b.0|\overline{b}.c.0 
\RewOver {\LabelSend {(a, b.0)} 1 | \bullet} 
b.0|\overline{b}.c.0 
\RewOver {\LabelSend {(b, 0)} 1 | \LabelSend {(\overline{b}, c.0)} 1}
0|c.0 
\RewOver {\bullet | \LabelSend {(c, 0)} 1 }
0|0 
$\\
In the first computation $\LabelSend {(a, b.0)} 1 | \bullet \ValPar
\bullet | \LabelSend {(\overline{b}, c.0)} 1$, hence the two
actions can be reversed in any order.  In the second computation
neither $\LabelSend {(a, b.0)} 1 | \bullet \ValPar \LabelSend {(b, 0)}
1 | \LabelSend {(\overline{b}, c.0)} 1$ nor $\LabelSend {(b, 0)} 1
| \LabelSend {(\overline{b}, c.0)} 1 \ValPar \bullet | \LabelSend
{(c, 0)} 1$ hold, hence the actions are necessarily undone in reverse
order. These two behaviors agree with both the standard notion of
concurrency in CCS, and with the behaviors of the causal-consistent
reversible extensions of CCS in the
literature~\cite{rccs,revUlidowski}. Indeed we conjecture that our semantics and the ones in the literature are equivalent.
\end{example}
More generally, in $(P | Q)$ reductions of $P$ and of $Q$ are
concurrent as expected, and the choice of reducing first $P$, then $Q$
or the opposite has no impact: if $P \RewOver u P'$ and $Q \RewOver v
Q'$, then we can permute $(P | Q) \RewOver {(u | \bullet)} (P' | Q)
\RewOver {(\bullet | v)} (P' | Q')$ into $(P | Q) \RewOver {(\bullet |
  v)} (P | Q') \RewOver {(u | \bullet)} (P' | Q')$.  Also, e.g., in
the rule for right parallel composition,
the label $(u | \bullet)$ is independent from $Q$.
For this reason we can permute the order of two concurrent
transitions without changing the labels.

Labels can be seen as derivation trees in the original CCS with some
information removed.  Indeed, for every derivation rule in CCS, there
is a production in the grammar of the labels. When extracting the
labels from the derivation trees:
\begin{itemize}

\item We must keep enough information from the original derivation
  tree to preserve determinism and co-determinism, and to define the
  $\ValPar$ relation.

\item We must remove enough information so that we can permute
concurrent transitions without changing the labels.
For example, in the label $(u | \bullet)$, there is no information on
the term which is not reduced.



\item Our labels and transition rules are close to the ones of the
  causal semantics of Boudol and Castellani~\cite{Boudol89}. Actually,
  our labels are a refinement of the ones of Boudol and Castellani: we
  need this further refinement since their transitions are not
  co-deterministic.
\end{itemize}
If we chose to take the whole derivation tree as a label, we would
have the same problem as in Remark~\ref{rmk:NaiveRefineLabel}: we
would have determinism and co-determinism, but we could not have a
definition of $\ValPar$ capturing the concurrency model of CCS.


\section{Examples with X-machines}
\label{sec:Automata}
X-machines \cite{xmachines} (also called Eilenberg machines) are a model of
computation that is a generalization of the concept of automaton.
Basically, they are just automata where transitions are relations over a set
$D$.
The set $D$ represents the possible values of a memory, and transitions modify
the value of the memory.
\begin{definition}[X-machines]
\mbox{}\\
An \emph{X-machine} on a set $D$ is a tuple $\mathcal{A} = (Q, I, F, \delta)$
such that:
\begin{itemize}
\item $Q$ is a finite set of states.
\item $I$ and $F$ are subsets of $Q$, representing initial and final states.
\item $\delta$ is a finite set of triplets $(q, \alpha, q')$ such that $q,
q' \in Q$ and $\alpha$ is a binary relation on $D$, defining the
transitions of the X-machine.
\end{itemize}
\end{definition}
The semantics of an X-machine is informally described as follows:
\begin{itemize}
\item The X-machine takes as input a value of $D$ and starts in an
initial state.

\item When the X-machine takes a transition, it applies the
relation to the value stored in the memory.

\item The value stored in the memory in a final state is the output.

\end{itemize}
This can be formalized using an LTS whose configurations are pairs
$(q,x)$ where $q$ is the state of the X-machine and $x$ the value of
the memory, and transitions are derived using the following inference
rule:
$$\infer{(q, x) \RewOver {\alpha} (q', y)}
{(q, \alpha, q') \in \delta \quad (x, y) \in \alpha}$$
X-machines are naturally a good model of sequential computations (both
deterministic and non-determi\-nistic).  In particular, a Turing machine
can be described as an X-machine~\cite{xmachines}.

X-machines have also been used as models of concurrency~\cite{comx}.
Below we will consider only a simple concurrent model:
several X-machines running concurrently and working on the same memory. 
This represents a set of sequential processes, one for each machine,
interacting using a shared memory.
We will start from the case where there are only two machines.
We will refine this model so that the refinement belongs to Theory~\ref{ax:Val}
and so that we can apply our framework.

First, we want to extend a single X-machine to make it reversible.
Notice that a single X-machine is a sequential model, hence at this
stage we have a trivial concurrency relation.
The LTS may be not deterministic and/or not co-deterministic both
because of the relation $\delta$, and because of the relation
$\alpha$. Hence, we will need to refine labels. For $\delta$, we can use the approach in Remark~\ref{rmk:NaiveRefineLabel}.
For actions $\alpha$, the approach is to split
each action $\alpha$ into a family of (deterministic and
co-deterministic) relations $(\alpha_i)_{i \in I}$ such that $\alpha =
\bigcup_{i \in I} \alpha_i$, and add to the label the index $i$ of the used $\alpha_i$.

\begin{definition}
Assume $D$ is a set, and $\alpha, \beta \in \mathcal{P}(D \times D)$.
We write $\alpha \ValPar \beta$ if and only if
$\alpha \circ \beta = \beta \circ \alpha$, where $\circ$ is the composition of
relations.
\end{definition}
\begin{definition}[Refined action]
We call a \emph{refined action} on $D$ an object $a$ such that:
\begin{itemize}

\item $\AHinst a$ is a set, representing the indices of the elements of the
splitting.

\item For all $i \in \AHinst a$, $a(i) \in \mathcal{P}(D \times D)$ with
$a(i)$ and $a(i)^{-1}$ functional relations
($a(i)$ is deterministic and co-deterministic).
\end{itemize}
\end{definition}
Notice that $\bigcup_{i \in \AHinst a} a(i)$ is indeed a relation on $D$.
Therefore, by forgetting the information added by the refinement
(how the action is split), a refined action
can be interpreted as a simple relation on $D$.

For instance, the following relations on $X^3$:
\[\begin{array}{c}
 \alpha = \{ ((x, y, z), (x, x, z)) \, \mid \, x, y, z \in X\} \quad
\beta = \{ ((x, y, z), (x, y, x)) \, \mid \, x, y, z \in X\}
\end{array}
\]
are not co-deterministic but can be refined as follows:
\begin{example}
Assume $X$ is a set and $D = X^3$.
We define the refined actions $a$ and $b$ on $D$ by:
\begin{itemize}

\item Setting $\AHinst a = \AHinst b = X$

\item For all $y \in X$, $a(y) \eqdef \{ ((x, y, z), (x, x, z)) ~ \mid ~
x, z \in X\}$.

\item For all $z \in X$, $b(z) \eqdef \{ ((x, y, z), (x, y, x)) ~ \mid ~
x, y \in X\}$.
\end{itemize}
Here $D$ represents a memory with three variables with values in $X$.
The action $a$ (resp. $b$) copies the first variable to the second
(resp. third) variable, indeed $\alpha = \bigcup_{y \in X} a(y)$ and
$\beta = \bigcup_{z \in X} b(z)$.

Notice that for all $y, z \in X$, $a(y) \ValPar b(z)$.
These actions $a$ and $b$ are indeed independent. Furthermore, when actions are permuted, the indices (here $y$ and $z$) remain the same:
$a(y) \circ b(z) = b(z) \circ a(y)$.
\end{example}
It is always possible to refine a given action by splitting
it into singletons:
For instance, the above actions $\alpha$ and $\beta$ can also be refined
as follows:
\begin{example}
We define the refined actions $a'$ and $b'$ as follows:
\begin{itemize}

\item $\AHinst {a'} \eqdef \AHinst {b'} = X^3$.

\item For all $x, y, z \in X$, $a'(x, y, z) \eqdef
\{ ((x, y, z), (x, x, z)) \}$.

\item For all $x, y, z \in X$, $b'(x, y, z) \eqdef
\{ ((x, y, z), (x, y, y)) \}$.
\end{itemize}
The refined action $a'$ (resp.\ $b'$) is another splitting of the action
$a$ (resp.\ $b$).

Unfortunately we generally do not have $a'(i) \circ b'(j) = b'(j) \circ a'(i)$.
Therefore $a$ (resp.\ $b$) and $a'$ (resp.\ $b'$) describe the same relation
$\alpha$ (resp.\ $\beta$) on $D$ but do not allow the same amount of concurrency.
Indeed, $a$ and $b$
allow one to define a non trivial $\ValPar$ while $a'$ and $b'$ do
not.
\end{example}
The two previous examples show that when refining an action,
the splitting must not be too thin to have a reasonable amount of concurrency.
In the examples above, $a$ and $b$ is a good refinement of $\alpha$ and $\beta$
but $a'$ and $b'$ is not.
The reason is the same as in
Remark~\ref{rmk:NaiveRefineLabel}.

Notice that, usually, we can refine any action that reads and writes
parts of the memory and the splitting must be indexed by the erased
information, if the original relation is not co-deterministic, and by
the created information, if the original relation is not
deterministic.

The next example shows how to model a simple imperative language
with a refined X-machine:

\begin{example}
An environment $\rho$ is a total map from an infinite set of variables
to $\mathds{N}$ such that the set of variables $x$ with $\rho(x) \neq
0$ is finite.  Let $D$ be the set of environments.
\begin{enumerate}
\item Assume that $x_1$, \ldots, $x_n$, $y$ are variables and $f$ is a total map
from $\mathds{N}^n$ to $\mathds{N}$.
Let $\alpha$ be the action on $D$ defined as follows:
\[\alpha \eqdef \{ (\rho, \extenv \rho y {f(\rho(x_1), \ldots, \rho(x_n))} \mid
\rho \in D\}\]
Then, $\alpha$ can be refined by defining the action $a$ as follows:
\begin{itemize}

\item $\AHinst a \eqdef \mathds{N}$

\item For all $v \in \AHinst a$,
$a(v) \eqdef \{ (\rho, \rho') \mid
\rho(y) = v \wedge \rho' = \extenv \rho y
{f(x_1, \ldots, x_n)}\}$.
\end{itemize}
This refined action is written $\AAssign y f {x_1, \ldots x_n}$.
When $f$ is injective, we do not need to refine $\alpha$:
It is already deterministic and co-deterministic.

\item Similarly, we can define the action $y +\!\!= f(x_1, \ldots,
  x_n)$ when for all $i$, $x_i \neq y$. This is the form of assignment
  used by Janus~\cite{Yoko07}, a language which is naturally
  reversible, and where reversibility is ensured by restricting the
  allowed constructs w.r.t.\ a conventional language. Indeed, we can
  see that the corresponding relation is deterministic and
  co-deterministic.

\item Assume that $x_1$, \ldots, $x_n$ are variables and $u$ is a subset of
$\mathds{N}^n$.
Let $\alpha$ be the action defined as follows:
\[\alpha \eqdef \{ (\rho, \rho) \mid \rho \in D \wedge (\rho(x_1), \ldots
\rho(x_n)) \in u\}\]
This action is used to create a branching instruction.
Relation $\alpha$ is already deterministic and co-deterministic, hence we do not
need to refine it.
It is written $\ATest {x_1, \ldots x_n} u$.
\end{enumerate}
Then, we can define the $\ValPar$ relation as usual: two actions are dependent
if there is a variable that both write, or that one reads and one
writes, independent otherwise. The functions $\text{rv}$ and $\text{wv}$
below compute the sets of read variables and of written variables,
respectively.
\[\begin{array}{rclrcl}

\ReadVar {\AAssign y f {x_1 \ldots x_n}} &\eqdef& \{x_1, \ldots x_n\} \qquad &
\WriteVar {\AAssign y f {x_1 \ldots x_n}} &\eqdef& \{y\}\\
\ReadVar {y +\!\!= f(x_1, \ldots, x_n)} &\eqdef& \{x_1, \ldots x_n\} \qquad &
\WriteVar {y +\!\!= f(x_1, \ldots, x_n)} &\eqdef& \{y\}\\
\ReadVar {\ATest {x_1 \ldots x_n} u} &\eqdef& \{x_1, \ldots x_n\} \qquad &
\WriteVar {\ATest {x_1, \ldots x_n} u} &\eqdef& \emptyset
\end{array}
\]
We have $a \ValPar b$ if and only if all the following conditions are satisfied:\\
$\ReadVar a \cap \WriteVar b = \emptyset \qquad$
$\ReadVar b \cap \WriteVar a = \emptyset \qquad$
$\WriteVar a \cap \WriteVar b = \emptyset$\\
%
We can then check that if $a \ValPar b$, then $a \circ b = b \circ a$.

\end{example}

Now we can define a refined X-machine, suitable for reversibility:

\begin{definition}
A \emph{refined X-machine} on $D$ is $\mathcal{A} = (Q, I, F, \delta)$ such
that:
\begin{itemize}

\item $Q$ is a finite set.

\item $I$ and $F$ are subsets of $Q$.

\item $\delta$ is a finite set of triplets $(q, a, q')$ such that
$q, q' \in Q$ and $a$ is a refined action.
\end{itemize}
\end{definition}
If we forget the refinement of the action we exactly have an
X-machine. In the semantics, each label would contain both the used
action $a$ and the index of the element of the split which is used.

We can now build systems composed by many X-machines interacting using
a shared memory.
\begin{example}

\label{ex:A2}

Assume we have two X-machines $\mathcal{A}_1 = (Q_1, I_1, F_1,
\delta_1)$ and $\mathcal{A}_2 = (Q_2, I_2, F_2, \delta_2)$ on $D$. We
want to describe a model composed by the two X-machines, acting on a
shared memory.

Terms $M$ are of the form $(q_1, q_2, x)$ where $q_1 \in Q_1$ is the
current state of the first X-machine, $q_2 \in Q_2$ is the current
state of the second X-machine and $x \in D$ is the value of the memory.  
Labels $u$ for the transitions are of the form $(k, q, a, q', i)$ with
$k \in \{1, 2\}$, $(q, a, q') \in \delta_k$ and $i \in \AHinst a$.
Here $k$ indicates which X-machine moves, $q$, $a$ and $q'$ indicate which
transition the moving X-machine performs, and
$i$ indicates which part of the relation is used.

The relation $\ValPar$ is defined as follows:
\[(k, q_1, a, q_1', i) \ValPar (k', q_2, b, q_2', j) \qquad \text{iff}
\qquad k \neq k' \wedge a(i) \ValPar a(j)\]
It means that two steps are independent if and only if they are performed
by two distinct X-machines and the performed actions are independent.

Transitions are defined by the following rules:
\[\begin{array}{cc}

\infer{(q_1, q_2, x) \RewOver {(1, q_1, a, q_1', i)} (q_1', q_2, y)}
{(q_1, a, q_1') \in \delta_1 \quad i \in \AHinst a \quad
(x, y) \in a(i)} \quad &

\infer{(q_1, q_2, x) \RewOver {(2, q_2, a, q_2', i)} (q_1, q_2', y)}
{(q_2', a, q_2') \in \delta_2 \quad i \in \AHinst a \quad
(x, y) \in a(i)}
\end{array}
\]

Then, we have the objects and properties of Theory~\ref{ax:Val}.
In particular, we have determinism and co-determinism, because in the
label $u = (k, q, a, q', i)$, $a(i)$ is deterministic and co-deterministic.

\end{example}

\begin{example}

\label{ex:AN}

Example~\ref{ex:A2} can be generalized to $n$ refined X-machines
(the terms $M$ being of the form $(q_1, \ldots, q_n, x)$).

\end{example}

\begin{example}

\label{ex:A1}

By adding to the model in Example~\ref{ex:AN} the restriction
``All the X-machines are equal''
we do not lose any expressiveness.
This will be relevant for the next example.
\iftoggle{tech_report}{
\begin{toappendix}[This is shown in Appendix.]

\noindent
\textbf{Details of the result in Example~\ref{ex:A1}}\\
We illustrate the reduction of the general case to the case where all
the $n$ X-machines are equal in the case $n = 2$ (which corresponds to
Example~\ref{ex:A2}).  The generalization for $n$ X-machines is
straightforward.
\begin{itemize}

\item Let $D' \eqdef \{0, 1, 2\} \times \{0, 1, 2\} \times D$.

\item Let $Q \eqdef \{i_0\} \cup Q_1 \cup Q_2$
(disjoint union).

\item Let $I \eqdef \{i_0\}$.

\item Let $F \eqdef F_1 \cup F_2$.

\item For each refined action $a$ on $D$,
let $\IntA a$ be the refined action on $D'$ defined by:
\begin{itemize}

\item $\AHinst {\IntA a} \eqdef \AHinst a$.

\item For all $i \in \AHinst a$,
$\IntA a (i) \eqdef \{((k_1, k_2, x), (k_1, k_2, y)) ~ \mid ~
k_1, k_2 \in \{1, 2\} \wedge k_1 \neq k_2 \wedge (x, y) \in a(i)\}$.

\end{itemize}

\item If $k \in \{1, 2\}$, then $\TakeRole k$ is the refined action on $D'$
defined by:
\begin{itemize}

\item $\AHinst {\TakeRole k} \eqdef \{1, 2\}$.

\item $\TakeRole k (1) \eqdef \{ ((0, k', x), (k, k', x)) ~ \mid ~
k' \in \{0, 1, 2\} \wedge k' \neq k \wedge x \in D\}$.

\item $\TakeRole k (2) \eqdef \{ ((k', 0, x), (k', k, x) ~ \mid ~
k' \in \{0, 1, 2\} \wedge k' \neq k \wedge x \in D\}$.

\end{itemize}

\item Let $\delta$ be a transition relation where elements are either:
\begin{itemize}

\item $(i_0, \TakeRole k, q)$ with $k \in \{1, 2\}$ and $q \in I_k$.

\item $(q, \IntA a, q')$ with $(q, a, q') \in \delta_k$ and $k \in \{1, 2\}$.

\end{itemize}

\end{itemize}

This idea can be illustrated with a theater play comparison:
\begin{itemize}

\item Every actor can play any role.

\item Two different actors cannot play the same role.

\item The play can only start when each role has been attributed to an actor.

\item Even if an actor could have played any role, once he has chosen a role,
he cannot change it.

\end{itemize}

\end{toappendix}
}{
This is shown in the companion technical report~\cite{TR}.
}

\end{example}

\begin{example}

\label{ex:AInf}

If the set of initial states of each X-machine is a singleton, we can
generalize Example~\ref{ex:A1} to a potentially infinite number of
X-machines, where however only a finite amount of them are not in their
initial state. However, an unbounded number of X-machines may have moved.

More formally, if we have a refined
X-machine $\mathcal{A} = (Q, \{ i_0\}, F, \delta)$.

\begin{itemize}

\item The labels $u$ are of the form $(k, q, a, q', i)$ with
$k \in \mathds{N}$, $(q, a, q') \in \delta$, and $i \in \AHinst a$.

\item The terms $M$ are of the form $(f, x)$ with $x \in D$ and
$f$ a total map from $\mathds{N}$ to $Q$ such that the set of
$k \in \mathds{N}$ with $f(k) \neq i_0$ is finite.

\item $\ValPar$ and $\RewOver u$ are defined similarly to their respective
counterpart in Example~\ref{ex:A2}.

\end{itemize}

The objects above satisfy the properties of Theory~\ref{ax:Val}.

\end{example}

In all the previous examples, as for CCS, we can apply the framework
of Section~\ref{sec:Formal} to define a causal-consistent reversible
semantics and have for free Loop lemma and causal consistency.

Example~\ref{ex:AInf} allows one to simulate the creation of new
processes dynamically.
Moreover, since we have no limitation on $D$, we can choose $D$ so to represent
an infinite set of communication channels.
We can also add the notion of synchronization between two X-machines.
Therefore we conjecture that by extending Example~\ref{ex:AInf} we
could get a reversible model as expressive as the $\pi$-calculus.

\section{Conclusion and Future Work}\label{sec:Concl}

We have presented a modular way to define causal-consistent reversible
extensions of formalisms as different as CCS and X-machines. This
contrasts with most of the approaches in the
literature~\cite{rccs,rhopi,revpi,revMuOz,GiachinoLMT15}, where specific
calculi or languages are considered, and the technique is heavily tailored to the
chosen calculus. However, two approaches in the literature are more
general. \cite{revUlidowski} allows one to define causal-consistent
reversible extensions of calculi in a subset of the path format.  The
technique is fully automatic. Our technique is not, but it can
tackle a much larger class of calculi.
In particular, X-machines do not belong to the path format since their terms
include an element of the set of values $X$, and $X$ is arbitrary.
\cite{Sobocinski} presents a
categorical approach that concentrates on the interplay between
reversible actions and irreversible actions, but provides no results
concerning reversible actions alone.

As future work we plan to apply our approach to other formalisms,
starting from the $\pi$-calculus, and to draw formal comparisons
between the reversible models in the literature and the corresponding
instantiations of our approach. We conjecture to be able to prove the
equivalence of the models, provided that we abstract from syntactic
details. A suitable notion of equivalence for the comparison is barbed
bisimilarity.
Finally, we could also show that the construction from terms to
reversible configurations given in Section~\ref{sec:Formal} is
actually monadic and that the algebras of this monad are also
relevant, since they allow one to inject histories into terms and make
them reversible.

\bibliographystyle{eptcs}
\bibliography{biblio}

\iftoggle{tech_report}{
\newpage
\appendix 

\input{proofs}
\input{toappendix}
}{
}

\end{document}